\documentclass[11pt]{article}
\usepackage[english]{babel}
\usepackage{graphicx}
\usepackage{caption}
\usepackage{subcaption}
\usepackage{scrextend}
\usepackage{thmtools}
	\usepackage{a4,geometry}
\usepackage{graphicx}
\usepackage{amsmath,amssymb,amsthm,mathtools}
\usepackage{paralist}
\usepackage{bm}
\usepackage{xspace}
\usepackage{url}
\usepackage{fullpage, prettyref}
\usepackage{boxedminipage}
\usepackage{wrapfig}
\usepackage{ifthen}
\usepackage{color}
\usepackage{xcolor}
\usepackage{framed}
\usepackage{algorithmic,algorithm}
\usepackage[pagebackref,letterpaper=true,colorlinks=true,pdfpagemode=none,urlcolor=blue,linkcolor=blue,citecolor=violet,pdfstartview=FitH]{hyperref}
\usepackage{fullpage}

\usepackage{thm-restate}
\newtheorem{theorem}{Theorem}[section]

\newtheorem{lemma}[theorem]{Lemma}
\newtheorem{claim}[theorem]{Claim}

\newtheorem{definition}[theorem]{Definition}

\newcommand{\ignore}[1]{}

\newcommand{\cE}{{\cal E}}
\newcommand{\cF}{\mathcal{F}}
\newcommand{\cG}{\mathcal{G}}

\newcommand{\cP}{\mathcal{P}}

\newcommand{\cW}{{\cal W}}

\newcommand{\R}{\mathbb R}

\newcommand{\eps}{\varepsilon}

\newcommand{\poly}{\mathrm{poly}}

\newcommand{\bw}{\boldsymbol{w}}

\newcommand{\bp}{\boldsymbol{p}}

\newcommand{\bu}{\boldsymbol{u}}
\newcommand{\bv}{\boldsymbol{v}}
\newcommand{\bx}{\boldsymbol{x}}
\newcommand{\by}{\boldsymbol{y}}

\newcommand{\bq}{\boldsymbol{q}}

\newcommand{\bP}{\boldsymbol{P}}
\newcommand{\bQ}{\boldsymbol{Q}}
\newcommand{\bR}{\boldsymbol{R}}

\newcommand{\NN}{\mathbb{N}}
\newcommand{\RR}{\mathbb{R}}

\newcommand{\EX}{\hbox{\bf E}}

\newcommand{\otilde}{\widetilde{O}}

\newcommand{\Sec}[1]{\hyperref[sec:#1]{\S\ref*{sec:#1}}} %section
\newcommand{\Eqn}[1]{\hyperref[eq:#1]{(\ref*{eq:#1})}} %equation
\newcommand{\Fig}[1]{\hyperref[fig:#1]{Fig.\,\ref*{fig:#1}}} %figure
\newcommand{\Tab}[1]{\hyperref[tab:#1]{Tab.\,\ref*{tab:#1}}} %table
\newcommand{\Thm}[1]{\hyperref[thm:#1]{Theorem\,\ref*{thm:#1}}} %theorem
\newcommand{\Fact}[1]{\hyperref[fact:#1]{Fact\,\ref*{fact:#1}}} %fact
\newcommand{\Lem}[1]{\hyperref[lem:#1]{Lemma\,\ref*{lem:#1}}} %lemma
\newcommand{\Prop}[1]{\hyperref[prop:#1]{Prop.~\ref*{prop:#1}}} %property
\newcommand{\Cor}[1]{\hyperref[cor:#1]{Corollary~\ref*{cor:#1}}} %corollary
\newcommand{\Conj}[1]{\hyperref[conj:#1]{Conjecture~\ref*{conj:#1}}} %conjecture
\newcommand{\Def}[1]{\hyperref[def:#1]{Definition~\ref*{def:#1}}} %definition
\newcommand{\Alg}[1]{\hyperref[alg:#1]{Alg.~\ref*{alg:#1}}} %algorithm
\newcommand{\Ex}[1]{\hyperref[ex:#1]{Ex.~\ref*{ex:#1}}} %example
\newcommand{\Clm}[1]{\hyperref[clm:#1]{Claim~\ref*{clm:#1}}} %example
\newcommand{\Step}[1]{\hyperref[step:#1]{Step~\ref*{step:#1}}} %example
\newcommand{\Cond}[1]{\hyperref[cond:#1]{Condition~\ref*{cond:#1}}}

\newcommand{\Akash}[1]{{\color{black!60!green} Akash: #1}}

\newcommand{\anchor}{{\tt Anchor}}
\newcommand{\isminorfree}{{\tt IsMinorFree} }
\newcommand{\localsearch}{{\tt LocalSearch}}
\newcommand{\estclip}{{\tt EstClip}}
\newcommand{\accept}{\textsf{ACCEPT}\xspace}
\newcommand{\reject}{\textsf{REJECT}\xspace}

\newcommand{\len}{\ell}

\newcommand{\wpr}[3]{p_{#1, #2}(#3)}
\newcommand{\wvec}[2]{\mathbf{p}_{#1, #2}}
\newcommand{\hwvec}[2]{\widehat{\bp}_{#1, #2}}

\newcommand{\qrwvec}[3]{\mathbf{q}_{[#3],#1,#2}}

\newcommand{\walkdist}[1]{\cW_#1}
\newcommand{\setOfSuchThat}[2]{ \left\{\; #1 \;\colon\; #2\; \right\} }

\newcommand{\RS}{\texttt{KKR}}

\newcommand{\tS}{\widetilde{S}}

\newcommand{\prw}[3]{p_{#1, #3}(#2)}

\newcommand{\prwprvec}[2]{\mathbf{p'}_{#1, #2}}
\newcommand{\clip}[2]{\mathrm{cl}(#1,#2)}
\newcommand{\hclip}[2]{{\widehat{cl}}(#1, #2)}

\newcommand{\projw}[2]{\boldsymbol{\tau}_{#1, #2}}
\newcommand{\projwp}[3]{{\tau}_{#1, #3}(#2)}
\newcommand{\projp}[3]{\hat{p}_{#1, #2}(#3)}

\newcommand{\runtime}{O\left(\eps^{-42 \cdot 2^{200r^4}}\right)}

\begin{document}
	\title{Random walks and forbidden minors II: A $\poly(d\eps^{-1})$-query tester for minor-closed properties of bounded degree graphs}
\author{Akash Kumar\thanks{Department of Computer Science, Purdue University. {\href{mailto:akumar@purdue.edu}{akumar@purdue.edu}} (Supported by NSF CCF-1618981.)}
\and C. Seshadhri\thanks{Department of Computer Science, University of California, Santa Cruz. {\href{mailto:sesh@ucsc.edu}{sesh@ucsc.edu}}
(Supported by NSF TRIPODS grant CCF-1740850 and NSF CCF-1813165)}
\and Andrew Stolman\thanks{Department of Computer Science, University of California, Santa Cruz. { \href{mailto:astolman@ucsc.edu}{astolman@ucsc.edu}} (Supported by NSF TRIPODS grant CCF-1740850)}
}

\begin{titlepage}

\date{}
\maketitle
\begin{abstract}
Let $G$ be a graph with $n$ vertices and maximum degree $d$. Fix some minor-closed property $\mathcal{P}$ (such as planarity).
We say that $G$ is $\varepsilon$-far from $\mathcal{P}$ if one has to remove $\varepsilon dn$ edges to make it have $\mathcal{P}$.
The problem of property testing $\mathcal{P}$ was introduced in the seminal work of Benjamini-Schramm-Shapira (STOC 2008)
that gave a tester with query complexity triply exponential in $\varepsilon^{-1}$.
Levi-Ron (TALG 2015) have given the best tester to date, with a quasipolynomial (in $\varepsilon^{-1}$) query complexity.
It is an open problem to get property testers whose query complexity is $\poly(d\varepsilon^{-1})$,
even for planarity.

In this paper, we resolve this open question. For any minor-closed property,
we give a tester with query complexity $d\cdot \poly(\varepsilon^{-1})$. 
The previous line of work on (independent of $n$, two-sided) testers is primarily combinatorial.
Our work, on the other hand, employs techniques from spectral graph theory. This paper
is a continuation of recent work of the authors (FOCS 2018) analyzing random walk algorithms
that find forbidden minors.
% It can be seen
% as a continuation from a recent
% one-sided $n^{1/2+o(1)}$-query minor-freeness tester of Kumar-Seshadhri-Stolman (FOCS 2018).
\end{abstract}

\thispagestyle{empty}

\end{titlepage}

\section{Introduction}
% !TEX root = two-sided-tester.tex
\newcommand{\ras}{\rightarrow}

The classic result of Hopcroft-Tarjan gives a linear time algorithm for deciding planarity~\cite{HT74}.
As the old theorems of Kuratowski and Wagner show, planarity is characterized by the non-existence
of $K_5$ and $K_{3,3}$ minors~\cite{K30, W37}. The monumental graph minor theorem of Robertson-Seymour
proves that any property of graphs closed under minors can be expressed by the
non-existence of a finite list of minors~\cite{RS:12, RS:13, RS:20}. Moreover,
given a fixed graph, $H$, the property of being $H$-minor-free can be decided
in quadratic time~\cite{KKR:12}. Thus, any minor-closed property of graphs can be decided in quadratic time.

What if an algorithm is not allowed to read the whole graph?
This question was first addressed in the seminal result of Benjamini-Schramm-Shapira (BSS)
in the language of property testing~\cite{BSS08}. 
Consider the model of random access to a graph adjacency list, as 
introduced by Goldreich-Ron~\cite{GR02}. Let $G = (V,E)$ be a graph
where $V = [n]$ and the maximum degree is $d$.
We have random access to the list through \emph{neighbor queries}.
There is an oracle that, given $v \in V$ and $i \in [d]$,
returns the $i$th neighbor of $v$ (if no neighbor exists, it returns $\bot$).

For a property $\cP$ of graphs with degree bound $d$, the distance of $G$ to $\cP$
is the minimum number of edge additions/removals required to make $G$ have $\cP$,
divided by $dn$. We say that $G$ is $\eps$-far from $\cP$ if the distance to $\cP$
is more than $\eps$.  
A property tester for $\cP$ is a randomized procedure that takes as input (query access to) $G$ and a proximity parameter, $\eps > 0$.
If $G \in \cP$, the tester must accept with probability at least $2/3$. If $G$ is $\eps$-far from $\cP$,
the tester must reject with probability at least $2/3$. A tester is one-sided if it accepts $G \in \cP$ with probability $1$.

Let $\cP$ be some minor-closed property such as planarity. 
BSS proved the remarkable result that any such $\cP$ is testable
in time independent of $n$. Their query complexity
was triply exponential in $(d/\eps)$. Hassidim-Kelner-Nguyen-Onak
improved this complexity to singly exponential, introducing the novel concept of partition oracles~\cite{HKNO}.
Levi-Ron gave a more efficient analysis, proving the existence of testers
with query complexity quasi-polynomial in $(d/\eps)$~\cite{LR15}.  
For the special cases of outerplanarity
and bounded treewidth, $\poly(d/\eps)$ query testers are known~\cite{YI:15,EHNO11}.

It has been a significant
open problem to get $\poly(d/\eps)$ query testers for all minor-closed properties. 
In Open Problem 9.26 of Goldreich's recent book
on property testing, he states the ``begging question of whether [the query complexity bound of testing minor-closed
properties]
can be improved to a polynomial [in $1/\eps$]"~\cite{G17-book}. 
Even for classic case of planarity, this was unknown. 

In this paper, we resolve this open problem.
\begin{theorem} \label{thm:open} Let $\cP$ be any minor-closed property of graphs with degree bound $d$.
There exists a (two-sided) tester for $\cP$ that runs in $d^2\cdot \poly(\eps^{-1})$ time.
\end{theorem}

Thus, properties such as planarity, series-parallel graphs, 
embeddability in bounded genus surfaces, linkless embeddable, and bounded treewidth are all testable in time $d^2\cdot \poly(\eps^{-1})$.

By the graph minor theorem of Robertson-Seymour~\cite{RS:20},
\Thm{open} is a corollary of our main result for testing $H$-minor-freeness.
As alluded to earlier, for any minor-closed property
$\cP$, there exists a finite list of graphs $\{H_1, H_2, \ldots, H_b\}$ satisfying
the following condition. A graph $G$ is in $\cP$ iff for all $i \leq b$, $G$ does not
contain an $H_i$-minor. Let $\cP_{H_i}$ be the property of being $H_i$-minor-free.
The characterization implies that if $G$ is $\eps$-far from $\cP$, there exists $i \leq b$
such that $G$ is $\Omega(\eps)$-far from $\cP_{H_i}$. Thus, property testers for $H_i$-minor
freeness imply property testers for $\cP$ (with constant blowup in the proximity parameter).

Our main quantitative theorem follows.

\begin{theorem} \label{thm:main} There is an absolute constant
$c$ such that the following holds.
Fix a graph $H$ with $r$ vertices. The property
of being $H$-minor-free is testable in $d(r/\eps)^{c}$ queries and
$d^2(r/\eps)^{2c}$ time.
\end{theorem}

We stress that $c$ is independent on $r$. Currently, our value
of $c$ is likely more than $100$, and
we have not tried to optimize the exponent of $\eps$.
We believe that 
significant improvement is possible, even by just tightening
the current analysis.
It would be of significant interest to get a better bound, even for the case of planarity.

\subsection{Related work} \label{sec:related}

Property testing on bounded-degree graphs is a large topic, and we point the reader to Chapter 9
of Goldreich's book~\cite{G17-book}. Graph minor theory is immensely deep,
and Chapter 12 of Diestel's book is an excellent reference~\cite{R-book}. 
We will focus on the work regarding property testing of $H$-minor-freeness.

As mentioned earlier, this line of work started with Benjamini-Schramm-Shapira~\cite{BSS08}.
Their tester basically approximates the frequency of all subgraphs with radius $2^{1/\eps}$,
which leads to the large dependence in $d/\eps$.
Central to their result (and subsequent) work is the notion of \emph{hyperfiniteness}.
A hyperfinite class of graphs has the property that the removal of a small constant fraction
of edges leaves connected components of constant size. Hassidim-Kelner-Nguyen-Onak
design \emph{partition oracles} for hyperfinite graphs to get improved testers~\cite{HKNO,LR15}. These oracles are local procedures
that output the connected component that a vertex lies in, without explicit knowledge of any global partition. This is extremely challenging
as one has to maintain consistency among different queries. The final construction is
an intricate recursive procedure that makes $\exp(d/\eps)$ queries. Levi-Ron gave 
a significantly simpler and more efficient analysis leading to their query complexity
of $(d\eps^{-1})^{\log \eps^{-1}}$. Newman-Sohler show how partition oracles
lead to testers for any property of hyperfinite graphs~\cite{NS13}.

Given the challenge of $\poly(d\eps^{-1})$ testers
for planarity, there has been focus on other minor-closed properties.
Yoshida-Ito give such a tester for outerplanarity~\cite{YI:15}, which was subsumed by a $\poly(d\eps^{-1})$
tester by Edelman et al for bounded treewidth graphs~\cite{EHNO11}. Nonetheless,
$\poly(d\eps^{-1})$ testers for planarity remained open.

Unlike general (two-sided) testers,
one-sided testers for $H$-minor-freeness must have a dependence on $n$. BSS
conjectured that the complexity of testing $H$-minor-freeness (and specifically planarity)
is $\Theta(\sqrt{n})$. Czumaj et al~\cite{C14} showed
such a lower bound for any $H$ containing a cycle, and gave an $\otilde(\sqrt{n})$
tester when $H$ is a cycle. 
Fichtenberger-Levi-Vasudev-W\"{o}tzel give an $\otilde(n^{2/3})$ tester for $H$-minor-freeness
when $H$ is $K_{2,k}$, the $(k\times 2)$-grid or the $k$-circus graph~\cite{FLVW:17}. 
Recently, Kumar-Seshadhri-Stolman (henceforth KSS) nearly resolved the BSS conjecture
with an $n^{1/2+o(1)}$-query one-sided tester for $H$-minor-freeness~\cite{KSS:18}.
The underlying approach uses the proof strategy of the bipartiteness
tester of Goldreich-Ron~\cite{GR99}.

The body of work on two-sided (independent of $n$) testers
is primarily combinatorial. The proof of \Thm{main}
is a significant deviation from this line of work, and 
is inspired by the spectral graph theoretic methods in KSS. 
As we explain in the next section, we do not require the full machinery of KSS,
but we do follow the connections between random walk behavior and graph minors.
The tester of \Thm{main} is simpler than
those of Hassidim et al and Levi-Ron, who use recursive
algorithms to construct partition oracles~\cite{HKNO,LR15}.

\subsection{Main ideas} \label{sec:ideas}

Let us revisit the argument of KSS, that gives an $n^{1/2+o(1)}$-query one-sided
tester for $H$-minor-freeness. We will take great liberties with parameters,
to explain the essence. The proof of \Thm{main} is inspired by the approach
in KSS, but the proof details deviate significantly. We discover that the full
machinery is not required. But the main idea is to exploit connections between
random walk behavior and graph minor-freeness.

First, we fix a random walk length $\ell = n^{\delta} \gg 1/\eps$, 
for small constant $\delta > 0$. One of the building blocks is a random walk procedure
that finds $H$-minors by performing $\sqrt{n}\cdot\poly(\ell)$ random walks of length $\ell$.
For our purposes, it is not relevant what the algorithm is, and we simply refer to this as the ``random walk procedure".
One of the significant concepts in KSS is the notion of a \emph{returning random walk}. For any subset of vertices $S \subset V$, an $S$-returning random walk of length $\ell$ is a random walk that starts from $S$ and ends at $S$. For any vertex $s \in S$, we use $\qrwvec{s}{\ell}{S}$ to denote the $|S|$-dimensional vector of probabilities of an $S$-returning walk of length $\ell$ starting from $s$.

KSS proves the following
two key lemmas. We use $c$ to denote some constant that depends only on $H$.
\begin{asparaenum}
    \item Suppose there is a subset $S \subseteq V$, $|S| \geq n/\ell$, with the following property.
    For at least half the vertices $s \in S$, $\|\qrwvec{s}{\ell}{S}\| \leq \ell^{-c}$.
    Then, whp, the $\sqrt{n}\cdot\poly(\ell)$-time random walk procedure finds an $H$-minor.
    \item Suppose there is a subset $S \subseteq V$, $|S| \geq n/\ell$, with the following property.
    For at least half the vertices $s \in S$, $\|\qrwvec{s}{\ell}{S}\| > \ell^{-c}$.
    Then, for every such vertex $s$, there is a cut of conductance at most $1/\ell$ contained in $S$,
    where all vertices (in the cut) are reached with probability at least $1/\poly(\ell)$ by $\ell$-length $S$-returning walks from $s$.
\end{asparaenum}

To get a one-sided tester, we run the $\sqrt{n}\cdot\poly(\ell)$ random walk procedure.
If it does not find an $H$-minor, then the antecedent of the second part above is true for all $S$ such that $|S| \geq n/\ell$.
The consequent basically talks of local partitioning \emph{within} $S$, even
though random walks are performed in the whole graph $G$. The statement
is proven using arguments from local partitioning theorems of Spielman-Teng~\cite{ST12}.
By iterating the argument, we can prove the existence of a set of $\eps dn$ edges, whose removal breaks $G$
into connected components of size at most $\poly(\ell)$. Moreover, a superset
of any piece can be ``discovered" by performing $\poly(\ell)$ random walks (of length $\ell$)
from some starting vertex. Roughly speaking, each piece has a distinct starting vertex.
Thus, if $G$ was $\eps$-far from being $H$-minor-free,
an $\eps$-fraction (by size) of the pieces will contain $H$-minors. A procedure that picks $\poly(\ell)$ random vertices
(to hit the starting vertex of these pieces)
and runs $\poly(\ell)$ random walks of length $\ell$ will, whp, cover a subgraph that contains an $H$-minor.
We refer to this as the ``local search procedure", which runs in $\poly(\ell)$ time.

This sums up the KSS approach.
Observe that in the first case above, 
by the probabilistic method, we are guaranteed the existence of a minor.
Let us abstract out the argument as follows. Let $\bQ$ be the statement/condition:
there exists a subset $S \subseteq V$, $|S| \geq n/\ell$ such that for at least
half the vertices $s \in S$,  $\|\qrwvec{s}{\ell}{S}\| \leq \ell^{-c}$.
KSS basically proves the following lemmas, which we refer to subsequently as Lemma 1 and Lemma 2.
\begin{asparaenum}
    \item $\bQ \Rightarrow$ $G$ contains an $H$-minor.
    \item $\neg \bQ \Rightarrow$ If $G$ is $\eps$-far from being $H$-minor-free, the local search procedure finds
    an $H$-minor whp.
\end{asparaenum}

\medskip

We now have an approach to get a $\poly(\eps^{-1})$ tester. Suppose we could set the random walk
length $\ell$ to be $\poly(\eps^{-1})$. And suppose we could test the condition $\bQ$ in time $\poly(\eps^{-1})$.
We could then run local search on top of this, and get a bonafide tester.

A simple adaptation of proofs of both Lemma 1 and Lemma 2 run into some 
fundamental difficulties. The proof of Lemma 1 crucially requires $\ell$ to 
be $n^{\delta}$ (or at least $\Omega(\log n)$).
The existence of the minor is shown
through the success of the $\sqrt{n}\cdot \poly(\ell)$ random walk procedure. Constant length random walks cannot find an $H$-minor, even
if $G$ was $\Omega(1)$-far from being $H$-minor-free ($G$ could be a $3$-regular expander). 

\paragraph{From hyperfiniteness to $\ell = \poly(\eps^{-1})$.} We employ a different (and simpler)
approach to reduce the walk length. A classic result of Alon-Seymour-Thomas asserts
that any $H$-minor-free bounded-degree graph $G$ satisfies the following ``hyperfinite" decomposition:
for any $\alpha \in (0,1)$, we can remove an $\alpha$-fraction of the edges to get
connected components of size $O(\alpha^{-2})$. Let us set $\alpha = \poly(\eps)$ and the walk
length $\ell \ll 1/\alpha$. We can show that $\ell$-length random walks in $G$
encounter the removed edges with very low probability. By and large, the walks
behave as if they were performed on the decomposition. Thus, walks in $G$
are ``trapped" in the small components of size $O(\alpha^{-2})$. Quantitatively,
we can show that most vertices $s$, $|\wvec{s}{\ell}\|_2 \geq \poly(\eps)$.
(We use $\wvec{s}{\ell}$ to denote the random walk distribution with starting vertex $s$.)
By the contrapositive: if there are at least $\poly(\eps)$-fraction
of vertices $s$ such that $|\wvec{s}{\ell}\|_2 \leq \poly(\eps)$, then $G$
contains an $H$-minor. This is easily testable. We get a more convenient, $\poly(\eps^{-1})$-query
testable version of Lemma 1.

\paragraph{Clipped norms for local partitioning.} We can now express
our new condition $\neg \bQ$ as: for more than a $(1-\poly(\eps))$-fraction
of vertices $s$, $\|\wvec{s}{\ell}\|_2 \geq \poly(\eps)$. This is a \emph{weakening}
of the antecedent. Previously, the condition referred to returning walks,
which have smaller norm. Furthermore, the returning walks specifically reference $S$,
the set in which we are performing local partitioning. Thus, we have some conditions
on the behavior of random walks within $S$ itself, which is necessary
to perform the local partitioning. Our new condition only refers to the $l_2$-norms
of random walks in $G$.

The new condition appears to be too fragile to get local partitioning within $S$.
It is possible that the $l_2$-norm of $\wvec{s}{\ell}$ is dominated
by a few vertices outside of $S$, whose $l_1$-norm is tiny. In other words,
an event of small probability dominates the $l_2$-norm. The existing proof of
Lemma 2 (from KSS) is not sensitive enough to handle such situations.

We overcome this problem by using a more robust version of norm,
called the \emph{clipped norm}. We define $\clip{\bx}{\xi}$
for distribution vector $\bx$ and $\xi \in (0,1)$ to be the smallest
    $l_2$-norm obtained by removing $\xi$ probability mass ($l_1$-norm) from $\bx$.
In other words, we can measuring the $l_2$-norm after ``clipping" away $\xi$
probability worth of outliers. We can prove a version of Lemma 2 with a lower
bound of the clipped norm. We need to now rework Lemma 1 in terms of clipped norms.
This turns out to be relatively straightforward. 

\paragraph{Putting it all together.} Our final tester is as follows. The length $\ell$ is set to $\poly(\eps^{-1})$.
It picks some random vertices, and estimates the $l_2$-norm 
of \emph{clipped} probability vectors of $\ell$-length random walks from these vertices. 
If sufficiently many of them have ``small" ($\poly(\eps^{-1})$) norms, then the tester 
rejects. Otherwise, it runs $\poly(\eps^{-1})$ walks to find a superset of a low conductance cut. 
The tester employs some exact $H$-minor finding algorithm on the observed subgraph.

\section{The algorithm} \label{sec:alg}

In the algorithm and analysis, we will use the following notation.
\begin{itemize}
	\item {\bf Random walks} - Unless stated otherwise, we consider lazy random walks on graphs. If the walk is at a vertex, $v$, it transitions to each neighbor of $v$ with probability $1/2d$ and remains at $v$ with probability $1 - \frac{d_v}{2d}$ where $d_v$ is the degree of the vertex v. Note that the stationary distribution is uniform.
    We use $M$ to denote the transition matrix of this random walk.

	\item $\wvec{v}{t}$ - the $n$-dimensional probability vector, where the $u$th entry is the probability that a length $t$ random walk started from $v$ ends at ${u}$.
    We denote each entry as $\wvec{v}{t}(u)$.

	\item $\| \cdot \|_p$ - the usual $l_p$ norm on vectors.

\end{itemize}

The two parameters to the algorithm are $\eps \in [0, 1/2]$, and a graph $H$ on $r \geq 3$ vertices. 
We set the walk length $\ell = \alpha r^3 + \lceil\eps^{-20}\rceil$, where $\alpha$ is some absolute constant. 

Our algorithm runs as a subroutine the exact quadratic time minor-finding algorithm of Kawarabayashi-Kobayashi-Reed~\cite{KKR:12}. We denote this procedure by \RS.

\medskip
\noindent
\fbox{
	\begin{minipage}{0.9\textwidth}
		{\isminorfree$(G, \eps, H)$}
		
		\smallskip
		\begin{compactenum}
		\item  Pick multiset $S$ of $\len^{21}$ uniform random vertices. \label{step:S}
			\item For every $s \in S$, run \estclip$(s)$ and \localsearch$(S)$.
            \item If any call to \localsearch{} returns FOUND, \reject. 
            \item If more than $2\len^{20}$ calls to \estclip{} return LOW, \reject.
			\item \accept	
		\end{compactenum}
\end{minipage}}

\noindent
\fbox{
	\begin{minipage}{0.9\textwidth}
	{$\localsearch(s)$}
	
	\smallskip
	\begin{compactenum}
		\item Perform $\ell^{21}$ independent random walks 
				of length $\ell^{11}$ from $s$. Add all the 
				vertices encountered to set $B_s$. \label{step:walk}
		\item Determine $G[B_s]$, the subgraph induced by $B_s$.
		\item If $\RS{(G[B_s],H)}$ finds an $H$-minor, return FOUND.
	\end{compactenum}
\end{minipage}}

\noindent
\fbox{
	\begin{minipage}{0.9\textwidth}
	{\estclip$(s)$}
	
	\smallskip
	\begin{compactenum}
		\item Perform $w = \len^{14}$ walks of length $\len$ from $s$.
		\item For every vertex $v$, let $w_v = \text{ number of walks that end at } v$.
        \item Let $T = \{v \ | \ w_v \geq \len^7/2\}$.
        \item If $\sum_{v \in T} w_v \geq w/3$, output HIGH, else output LOW.
	\end{compactenum}
\end{minipage}}

\medskip

\Thm{main} follows directly from the following theorems.

\begin{theorem} \label{thm:completeness}
		If $G$ is $H$-minor-free, \isminorfree outputs \accept{} with probability at least $2/3$.
\end{theorem}

\begin{theorem} \label{thm:soundness}
		If $G$ is $\eps$-far from $H$-minor-freeness, then \isminorfree outputs \reject{} with probability
		at least $2/3$.
\end{theorem}

\begin{claim}\label{clm:runtime}
	There exists an absolute constant, $c$ such that the query complexity of \isminorfree{}
    is $O(d(r/\eps)^c)$ and time complexity is $O(d^2(r/\eps)^2)$.
\end{claim}

\begin{proof} The entire algorithm is based on performing $\poly(\ell)$ random walks
of length $\poly(\ell)$. Note that $\ell = \poly(r/\eps)$. The dependence
on $d$ appears because the subgraph $G[B_s]$ is constructed by query
the neighborhood of all vertices in $B_s$. The quadratic overhead in running time
is because of KKR.
\end{proof}

\section{Random walks do not spread in minor-free graphs}\label{sec:reduce}

We first define the clipped norm.

\begin{definition} \label{def:clip} Given $\bx \in (\RR^+)^{|V|}$ and parameter $\xi \in [0,1)$, the $\xi$-clipped
vector $\clip{\bx}{\xi}$ is the lexicographically least vector $\by$ optimizing the program:
$\min \|\by\|_2$, subject to $\|\bx - \by\|_1 \leq \xi$ and $\forall v \in V, \by(v) \leq \bx(v)$.
\end{definition}

The clipping operation removes ``outliers" from a vector, with the intention of minimizing the $l_2$-norm.
For a probability distribution $\wvec{s}{\ell}$, a small value of $\|\wvec{s}{\ell}\|^2_2$ is a measure
of the spread of the walk. But this is a crude lens. There may be one large coordinate in $\wvec{s}{\ell}$
that determines the norm, while all other coordinates are (say) uniform. The clipped
norm better captures (for our purposes) the notion of a random walk spreading.

We state the main result of this section. The constant $3/8$ below
is just for convenience, and can be replaced by any non-zero constant
(with a constant drop in the lower bound).

\begin{lemma}\label{lem:clip:new}
There is an absolute constant $\alpha$ such that the following holds.
Let $H$ be a graph on $r$ vertices. Suppose $G$ is a $H$-minor-free
graph. Then for any $\len \geq \alpha r^3$, there exists 
at least $(1 - 1/\len )n$ vertices such that
$\| \clip{\wvec{v}{\len}}{3/8}\|_2^2 \geq \len^{-7}$.
	\end{lemma}

In order to show this lemma, we will use the classic decomposition theorem for 
minor-free graphs by Alon-Seymour-Thomas~\cite{AST90}. It originally appears phrased 
in terms of a weight function $w: V \rightarrow \mathbb{R}^+$. We use 
the uniform weight function $\forall v\in V$ $w(v) = 1/n$ to obtain 
the restatement below.

\begin{lemma}[Proposition 4.1 of \cite{AST90}]\label{lem:ast}
There is an absolute constant $\alpha$ such that the following holds.
	Let $H$ be a graph on $r$ vertices. Suppose $G$ is an $H$-minor-free 
	graph with maximum degree $d$. Then, for all $k \in \NN$,
    there exists a set of at most  
	$\alpha n r^{3/2}/k^{1/2}$ vertices whose 
	removal leaves $G$ will all connected components of size at most $k$.
\end{lemma}

It is convenient to think of the Markov chain on $G$
in terms of a multigraph on $G$, with $2d$ edges from each vertex. Each
edge has probability exactly $1/2d$, and self-loops consist of many such edges.
Note that every edge of the original graph is a single edge in this multigraph.
For any subset of vertices $C \subseteq V$, let us define the random
walk restricted to $C$. We remove every cut edge $(u,v)$ (where $u \in C$
and $v \notin C$) and add a self-loop of the same probability at $u$. This produces a Markov chain on $C$
that is symmetric. Given a subset $C$ and $v \in C$, we use $\prwprvec{v}{t}$ to denote
the distribution of endpoints of $t$-length random walk starting from 
$v$ and restricted to $C$. (In our use, $C$ will apparent from context, so we will not
carry the dependence on $C$ in the notation.)

The following claim relates the clipped norms of the $\wvec{v}{t}$
and $\prwprvec{v}{t}$ vectors.

\begin{claim} \label{clm:good:vertices:new} Let $C \subset V$ and $v \in C$.
Let $\eta$ be the probability that a $t$-length random walk from $v$ (in $G$)
leaves $C$. For any $\sigma > \eta$, $\|\clip{\wvec{v}{t}}{\sigma - \eta}\|^2_2
\geq \|\clip{\prwprvec{v}{t}}{\sigma}\|^2_2$.
\end{claim}

% 		Let $C$ be some subgraph of $G$, and suppose 
% 		$\|\clip{\wvec{v}{t}}{1/2 - \eta}\|_2^2 \leq \alpha$ and that 
% 		the probability that a $t$-length walk from $v$ leaves $C$ is 
% 		at most $\eta$. Let $\prwprvec{v}{t}$ denote the probability 
% 		distribution vector for $t$-step random walks started from $v$ 
% 		restricted to $C$. Then $\|\clip{\prwprvec{v}{t}}{1/2}\|^2_2 \leq \alpha$. 
% 	\end{claim}

\begin{proof}
The random walk restricted to $C$ is obtained by adding some self-loops
that are not in the original Markov chain. Color all these self-loops
red. Let ${\bf r}_{v,t}(u)$ be the probability of a 
$t$-length walk from $v$ to $u$ that contains a red edge. Any path without 
a red edge is a path in $G$ (with the same probability), so 
$\prwprvec{v}{t}(u) \leq \wvec{v}{t}(u) + {\bf r}_{v,t}(u)$.

Note that $\sum_{u \in C} {\bf r}_{v,t}(u)$ is the total probability of
a random walk from $u$ restricted to $C$ encountering a red self-loop.
Red self-loops correspond to cut edges in the original graph,
and thus, this is the probability of encountering a cut edge.
Hence, $\sum_{u \in C} {\bf r}_{v,t}(u) \leq \eta$.

Intuitively, we can obtain a $\sigma$-clipping of $\prwprvec{v}{t}$
by first clipping at most $\eta$ probability mass to get $\wvec{v}{t}$,
and then performing a $(\sigma-\eta)$-clipping of $\wvec{v}{t}$.
We formalize this below.

Let $\bq = \clip{\wvec{v}{t}}{\sigma - \eta}$, and let us define the 
$|C|$-dimensional vector $\bw$ by $\bw (u) = \min \left( \bq(u), \prwprvec{v}{t}(u)\right)$. Since $\bw$ is non-negative and $\bw(u) \leq \bq(u)$ for all $u \in C$, 
it follows that $\|\bw\|_2^2 \leq \|\bq\|_2^2 = \|\clip{\wvec{v}{t}}{\sigma-\eta}\|^2_2$. 
      By construction, for all $u \in C$, $\bw(u) \leq \prwprvec{v}{t}(u)$.
We will prove that $\|\bw - \prwprvec{v}{t}\|_1 \leq \sigma$, implying 
that $\|\clip{\prwprvec{v}{t}}{\sigma}\|^2_2 \leq \|\bw\|^2_2$. This will
complete the argument.

Let $D \subseteq C$ be the set of coordinates such that $\bq(u) < \prwprvec{v}{t}(u)$.
Since $\bw(u) = \min(\bq(u),\prwprvec{v}{t}(u))$,
$\| \prwprvec{v}{t} - \bw \|_1 = \sum_{u \in D} [\prwprvec{v}{t}(u) - \bq(u)]$. Combining
with the previous observations and noting that $\bq = \clip{\wvec{v}{t}}{\sigma - \eta}$,
\begin{align}
   	\|\prwprvec{v}{t} - \bw \|_1 & \leq \sum_{u \in D}[\wvec{v}{t}(u) + {\bf r}_{v,t}(u) - \bq(u)] 
    \leq \|\wvec{v}{t}(u) - \bq\|_1 + \sum_{u \in C} {\bf r}_{v,t}(u) \leq (\sigma-\eta) + \eta = \sigma 
\end{align}
\end{proof}

% 
% 
% 
% 		It suffices to show that $\bw$ is a feasible solution to the 
% 		$\clip{\prwprvec{v}{t}}{1/2}$ program. That $\bw(u) \leq \prwprvec{v}{t}(u)$ 
% 		follows immediately from the definition. In order to check the $l_1$ 
% 		condition observe that the random walk restricted to $C$ is formed by 
% 		taking every edge leaving $C$ and replacing it with a self loop. Let us color 
% 		all such self loops red. 	    \begin{equation}\label{eq:restricted-walk:new}
% 	    	

We now prove the main lemma of this section.

\begin{proof}[Proof of \Lem{clip:new}]
	Fix some $\ell \in \NN$, $\ell > \alpha r^3$ and use \Lem{ast} with $k = r^3\ell^6$.
There exists a subset $R$ of at most $\alpha dn/\ell^3$ edges
whose removal breaks up $G$ into connected components of size at most $r^3 \ell^{6}$.
Refer to these as AST components.
Now, consider an $\ell$-length walk in $G$ starting from the stationary 
distribution (which is uniform). The probability that this walk encounters an edge
in $R$ at any step is exactly $|R|/2dn$. Let the random variable $X_v$ be 
the	number of edges of $R$ encountered in an $\ell$-length walk from $v$. Note that
when $X_v = 0$, then the walk remains in the AST component containing $v$.
Thus,
 $$(1/n) \sum_v \Pr[\textrm{walk from $v$ leaves AST component}]
 \leq \EX_{v \sim \text{u.a.r.}}[X_v] = \ell |R|/2dn \leq \alpha/(2 \ell^{2})$$
Since $\ell > \alpha r^3 > 4\alpha$, we can upper bound by $1/8\ell$.
By the Markov bound, for at least $(1-1/\ell)n$ vertices,
the probability that
an $\ell$-length walk starting at $v$ encounters an edge of $R$ and thus leaves 
the AST piece containing $v$ is at most $1/8$. Denote the set of these vertices 
by $S$. 
	
Consider any $s \in S$. Suppose it is contained in the AST component $C$.
Note that $\|\clip{\prwprvec{s}{\ell}}{1/2}\|_1
\geq 1/2$. Furthermore, it has support at most $|C| \leq r^3\ell^6$.
By Jensen's inequality, $\|\clip{\prwprvec{s}{\ell}}{1/2}\|^2_2 \geq 1(4r^3\ell^6)$.
As argued earlier, the probability that a random walk (in $G$) from $s$ leaves $C$
is at most $1/8$. Applying \Clm{good:vertices:new} for $\sigma = 1/2$ and $\eta = 1/8$,
we conclude that $\|\clip{\wvec{s}{\ell}}{1/2-1/8}\|^2_2 \geq 1/(4r^3\ell^6) \geq 1/\ell^7$.
(For convenience, we assume that $\alpha > 4$.)
% 
%  
%     Now let us turn our attention to the 
% 	walks from vertices in $S$ restricted to the AST piece, $C$, in which they started. 
% 	We will use $\prwprvec{v}{t}$ to denote the corresponding vectors. Observe 
% 	that for all $v \in S$, $\clip{\prwprvec{v}{\ell}}{1/2}$ is such that 
% 	$\|\clip{\prwprvec{v}{\ell}}{1/2}\|_1 = 1/2$, and they all have support of 
% 	size at most $r^3 \ell^{6}$. An application of Jensen's reveals that 
% 	$\| \clip{\prwprvec{v}{\ell}}{1/2}\|_2^2 \geq \frac{1}{4} r^{-3}\ell^{-6}$. By the property of $S$ and the
% 	contrapositive of \Clm{good:vertices:new}, we can then infer that 
% 	$\|\clip{\wvec{v}{\ell}}{1/2 - 1/8}\|_2^2 \geq \frac{1}{8} r^{-3}\ell^{-6} \geq \ell^{-7}$ for all vertices in $S$.
\end{proof}

\section{The existence of a discoverable decomposition}\label{sec:discoverable}

If many vertices have large clipped norms, we prove that $G$ can be partitioned
into small low conductance cuts. Furthermore, each cut can be discovered by
$\poly(\ell)$ $\ell$-length random walks. The analysis follows the structure
given in~\cite{KSS:18}.

\begin{restatable}{lemma}{partition} \label{lem:partitioning}
Let $c > 1$ be a parameter. Suppose there exists
$S \subseteq V$ such that $|S| > n/\len^{1/5}$ and $\forall s\in S, \|\clip{\wvec{s}{\len}}{1/4}\|^2_2 > \len^{-c}$.
Then, there exists $\tS \subseteq S$ with $|\tS| \geq |S|/4$ such that
for each $s \in \tS$, there exists a subset $P_s \subseteq S$ where
\begin{itemize}
	\item $\forall v \in P_s$, $\sum_{t < 16\len^{c+1}} \prw{s}{v}{t} \geq 1/8\len^{c+1}$.
	\item $|E(P_s, S \setminus P_s)| \leq 4d|P_s|\sqrt{c\len^{-1/5}\log \len}$.
\end{itemize} 
\end{restatable}

A straightforward application of this lemma leads to the main decomposition theorem.

\begin{theorem}\label{thm:partition}
Suppose there are at least $(1-1/\len^{1/5})n$ vertices $s$ such that
$\|\clip{\wvec{s}{\ell}}{1/4}\|^2_2 > \len^{-c}$. Then, there 
is a partition $\{P_1, P_2, \ldots, P_b\}$ of the vertices such that:
\begin{itemize}
    \item For each $P_i$, there exists $s \in V$ such that: $\forall v \in P_i$, $\sum_{t < 10\len^{c+1}} \prw{s}{v}{t} \geq 1/8\len^{c+1}$.
    \item The total number of edges crossing the partition is at most $8dn\sqrt{c\len^{-1/5}\log \len}$.
\end{itemize}
\end{theorem}

\begin{proof} We simply iterate over \Lem{partitioning}. 
Let $T = \{s \ | \ \|\clip{\wvec{s}{\ell}}{1/4}\|^2_2 \leq \len^{-c}\}$. By assumption, $|T| \leq n/\len^{1/5}$.
We will maintain a partition of the vertices $\{T, Q_1, Q_2, \ldots, Q_a, S\}$ with the 
following properties. (1) Each $Q_i$ satisfies the first condition of the theorem.
(2) The total number of edges crossing the partition is at most $4d\sqrt{c\len^{-1/5}\log \len}\sum_{i \leq a} |Q_i| + d|T|$.
We initialize with the trivial partition $\{T, S = V \setminus T\}$.

As long as $|S| > n/\len^{1/5}$, we invoke \Lem{partitioning}. We get a new set $Q \subseteq S$
satisfying the first condition of the theorem, and the number of edges from $Q$ to $S\setminus Q$
is at most $4d\sqrt{c\len^{1/5}\log \len}|Q|$. We add $Q$ to our partition, reset $S = S\setminus Q$,
and iterate.

When this process terminates, $|S| \leq n/\len^{1/5}$. We get the final partition
by removing all edges incident to $S \cup T$. Alternately, every single vertex
in $S \cup T$ becomes a separate set. Note that a single vertex trivially satisfies
the first condition of theorem, since for all $s$, $\wpr{s}{s}{1} \geq 1/2$.
The total number of edges crossing the partition is at most $4dn\sqrt{c\len^{-1/5}\log \len} + 2dn\len^{-1/5}
\leq 8dn\sqrt{c\len^{-1/5}\log \len}$.
\end{proof}

\subsection{Proving \Lem{partitioning}}\label{sec:local-partitioning}

An important tool used to argue about conductances within $S$ is the projected Markov chain. 
These ideas come from the work of Kale-Peres-Seshadhri to analyze random
walks in noisy expanders~\cite{KalePS:13}, and were used by the authors in their previous
paper on one-sided testers for minor-freeness~\cite{KSS:18}.
We closely follow the structure and notation of that paper, and explicitly mention
the differences.

We define the ``projection" of the random walk onto the set $S$.
We define a Markov chain $M_S$, over the set $S$. We retain all transitions
from the original random walk on $G$ that are within $S$, and we denote these by $e^{(1)}_{u, v}$ for every $u$ to $v$ transition in the random walk on $G$.
Additionally, for every $u,v \in S$ and $t \geq 2$,
we add a transition $e^{(t)}_{u,v}$. The probability of this transition
is equal to the total probability of $t$-length walks in $G$ from $u$
to $v$, where all internal vertices in the walk lie outside $S$.

Note that $e^{(t)}_{u,v} = e^{(t)}_{v,u}$. Since $G$ is irreducible and the stationary mass on $S$ is nonzero, all walks eventually reach $S$. 
Thus, for any $u$, $\sum_t \sum_v e^{(t)}_{u,v} = 1$,
so $M_S$ is a symmetric Markov chain. The stationary distribution of $M_S$
is uniform on $S$.
% 
% Thus the outgoing transition probabilities from each $v$ in $M_S$ sum to 1, and hence $M_S$ is a valid Markov chain. Since for any $t$ length walk from $u$ to $v$, all vertices on which fall outside $S$, we added
% a transition $e^{(t)}_{u,v}$, it can be seen that the transition probabilities for all edges which end at $v$ in $M_S$ sum to 1. Therefore the transition matrix of $M_S$ remains doubly stochastic, and the stationary distribution is uniform on $S$. 

For a transition $e^{(t)}_{u, v}$ in $M_S$, define the ``length" of this transition to be $t$. 
For clarity, we use ``hops" to denote the number of steps of a walk in $M_S$,
and retain ``length" for walks in $G$. The length of an $h$ hop random walk in $M_S$ is defined to be the sum of the lengths of the transitions it takes.

We use $\projw{s}{h}$ to denote the distribution of the $h$-hop walk from $s$,
and $\projwp{s}{v}{h}$ to denote the corresponding probability of reaching $v$.
We use $\walkdist{h}$ to denote the distribution of $h$-hop walks starting from the uniform distribution.

The following lemma is crucial for relating walks in $G$ with $M_S$.

\begin{lemma} [Lemma 6.4 of~\cite{KSS:18-full}] \label{lem:kac} 
	$\EX_{W\sim \walkdist{h}}[\textrm{length of $W$}] = hn/|S|$
\end{lemma}

We come to an important lemma. The conditions in \Lem{partitioning} are on the clipped norms of random walks in $G$,
but the conclusion (regarding the cut) refers to conductances within the projected Markov chain $M_S$.
The following lemma shows that random walks in $M_S$ must also be sufficiently trapped.
This is an analogue of Lemma 6.5 of~\cite{KSS:18-full}, but the proof deviates significantly
because of the use of clipped norms.

\begin{lemma}\label{lem:hop-infty-norm}
There exists a subset $S' \subseteq S$, $|S'| \geq |S|/2$, such that $\forall s \in S', \|\projw{s}{\len^{1/5}}\|_\infty \geq 1/2\len^{c+1}$.
\end{lemma}

\begin{proof} Consider $\len$-length random walks in $G$ starting from $s \in S$. For any such walk,
we can define the number of hops it makes as the number of vertices in $S$ encountered minus one.

For $h \in \NN$ and $s \in S$, define the event $\cE_{s,h}$ that an $\len$-length walk from $s$ makes $h$ hops.
We will further split this event into $\cF_{s,h}$, when the walk ends at $S$,
and $\cG_{s,h}$, when the walk does not end at $S$. 
A walk that ends in $S$ directly corresponds to an $h$-hop walk in $M_S$.
By \Lem{kac}, $|S|^{-1} \sum_{s \in S} \Pr[\cF_{s,h}]\len \leq hn/|S|$.
Consider any walk in the event $\cG_{s,h}$. If one continued until it ends in $S$,
this gives a walk in $M_S$ with a single additional hop (and a longer length).
Thus, the total probability mass $\Pr[\cG_{s,h}]$ corresponds to walks in $M_S$
that make $(h+1)$ hops and have length at least $\len$.
By \Lem{kac} again, $|S|^{-1} \sum_{s \in S} \Pr[\cG_{s,h}]\ell \leq (h+1)n/|S|$.

Summing these bounds and applying the size bound on $S$, 
$$|S|^{-1} \sum_{s \in S} \Pr[\cE_{s,h}]\len \leq (2h+1)n/|S| \leq \len^{1/5}(2h+1) \Longrightarrow
|S|^{-1} \sum_{s \in S} \Pr[\cE_{s,h}] \leq \len^{-4/5}(2h+1) $$

Now, we sum over $h$ and use the fact that $\len$ is a sufficiently large constant.
$$ |S|^{-1} \sum_{h \leq \len^{1/5}} \sum_{s \in S} \Pr[\cE_{s,h}] \leq \len^{-4/5} \sum_{h \leq \len^{1/5}} (2h+1)
\leq 4\len^{-2/5} < 1/10$$
% \Longrightarrow |S|^{-1} \sum_{s \in S} \sum_{h \leq \len^{1/5}} \Pr[\cE_{s,h}] \leq 4\len^{-2/5} < 1/10 $$
%
By the Markov bound, there is a set $S'$, $|S'| \geq |S|/2$ such that $\forall s \in S', \sum_{h \leq \len^{1/5}} \Pr[\cE_{s,h}] < 1/5$.

For $v \in V$,
let $y_s(v)$ be the probability that an $\len$-length walk from $s$ to $v$ makes at most $\len^{1/5}$ hops.
Note that $\sum_{v \in S} y_s(v) \leq \sum_{h \leq \len^{1/5}} \Pr[\cE_{s,h}] < 1/5$. 
We now use the clipped norm definition. Since $\|\clip{\wvec{s}{\len}}{1/4}\|^2_2 \geq \len^{-c}$,
$\sum_{v \in V} (\wpr{s}{\len}{v} - y_s(v))^2 \geq \len^{-c}$. This is important,
since we can ``remove" the low hop walks and still have a large norm.

Consider the probability $\alpha$ that a $2\len$-length walk from $s$ back to $s$ 
makes at least $\len^{1/5}$ hops. (Note that this corresponds to walks in $M_S$.)
Clearly, any walk going from $s$ to $v$ in an $\len$-length walk
making at least $\len^{1/5}$ hops and then returning to $s$ in an $\len$-length walk
contributes to this probability. Thus, we can lower bound $\alpha$
by $\sum_{v \in V}(\wpr{s}{\len}{v} - y_s(v))^2 \geq \len^{-c}$. 
Note that all walks considered make at most $2\ell$ hops.

Thus, $\sum_{h \geq \len^{1/5}}^{2\ell} \|\projw{s}{\len^{1/5}}\|_\infty \geq \len^{-c}$.
Since the infinity norm is non-increasing in hops, by averaging, $\|\projw{s}{\len^{1/5}}\|_\infty \geq 1/2\len^{c+1}$.
\end{proof}

The remaining proof of \Lem{partitioning} is almost identical to analogous calculations in Section 6 of~\cite{KSS:18-full}.
Therefore, we move it to the appendix.

\section{Proof of main result}\label{sec:main-results}

Before we show \Thm{completeness} and \Thm{soundness}, we argue about the guarantees
of \estclip. The proofs of the next two claims are relatively routine concentration arguments.
Recall that $T$ is the vertex set constructed in a call to \estclip$(s)$.

\begin{claim} \label{clm:set} Consider any vertex $s$. With probability at least $1-2^{-1/\eps^2}$
over the randomness in \estclip$(s)$: all $v$ such that $\wvec{s}{\ell}(v) \geq 1/\len^7$
are in $T$, and no $v$ such that $\wvec{s}{\ell}(v) \leq 1/\len^8$ is in $T$.
\end{claim}

\begin{proof} Consider $v$ such that $\wvec{s}{\ell}(v) \geq 1/\len^7$. 
Recall that the total number of walks is $w = \len^{14}$.
The expected
value of $w_v$ is at least $\len^{14}/\len^7 = \len^7$. Note that $w_v$ is a sum
of Bernoulli random variables. By a multiplicative Chernoff bound (Theorem 1.1 of~\cite{DuPa-book}),
$\Pr[w_v \leq \len^7/2] \leq \exp(-\len^7/8)$. There are at most $\len^7$ such vertices $v$.
By a union bound over all of them, the probability that some such $v$
is not in $T$ is at most $\len^7\cdot \exp(-\len^7/8) \leq \exp(-\len^6) \leq 2^{-2/\eps^2}$.
(Note that $\len > \eps^{-20}$.)
This proves the first part.

For the second part, consider $v$ such that $\wvec{s}{\ell}(v) \leq 1/\len^8$. 
We split into two cases.

{\bf Case 1, $\wvec{s}{\ell}(v) \geq \exp(-\len/2)$.} The expectation of $w_v$
is at most $\len^{14}/\len^8 = \len^6$. Since $\len^7/2 \geq 2e\len^6$,
by a Chernoff bound (third part, Theorem 1.1 of~\cite{DuPa-book}),
$\Pr[w_v \geq \len^7/2] \leq 2^{-\len^7/2}$. There are at most $\exp(\len/2)$
such vertices $v$. Taking a union bound over all of them,
the probability that any such vertex appears in $T$ is at most $\exp(\len/2)2^{-\len^7/2}
\leq 2^{-\len^5} \leq 2^{-2/\eps^2}$.

{\bf Case 2, $\wvec{s}{\ell}(v) < \exp(-\len/2)$.} For convenience, set $p = \wvec{s}{\ell}(v)$.
The probability that $w_v \leq 1$ is:
\begin{equation}
(1-p)^w + wp(1-p)^{w-1} \geq (1-wp) + wp(1-p(w-1)) = 1 - p^2w(w-1) \geq 1-p^2w^2
\end{equation}
(We use the inequality $(1-x)^r \geq 1-xr$, for $|x| \leq 1, r \in \NN$.)
Thus, the probability that $w_v > 1$ is at most $p^2w^2$. Note that $\len^7/2$ (the
threshold to be placed in $T$) is at least $2$.

Let us take a union bound over all such vertices. We note that $w = \len^{14}$ and $\len > \eps^{-20}$.
The probability that any such $v$ is placed in $T$ is at most 
\begin{equation}
\sum_{v: \wvec{s}{\ell}(v) < \exp(-\len/2)} \wvec{s}{\ell}(v)^2w^2 \leq \len^{28}\exp(-\len/2) \sum_{v} \wvec{s}{\ell}(v)
\leq \exp(-1/\eps^2)
\end{equation}
We union bound over all errors to complete the proof.

\end{proof}

We can now argue about the main guarantee of \estclip.

\begin{claim} \label{clm:estclip} For all vertices $s$, with probability at least $1 - 2^{-1/\eps}$
over the randomness of \estclip$(s)$:
\begin{asparaitem}
    \item If $\|\clip{\wvec{s}{\ell}}{1/4}\|^2_2 < \len^{-8}/400$, then \estclip$(s)$ outputs LOW.
    \item If $\|\clip{\wvec{s}{\ell}}{3/8}\|^2_2 > \len^{-7}$, then \estclip$(s)$ outputs HIGH.
\end{asparaitem}
\end{claim}

\begin{proof} Consider the first case.
Let $H = \{v \ | \wvec{s}{\ell}(v) \geq \len^{-8}\}$. We first argue
that $\sum_{v \in H} \wvec{s}{\ell}(v) \leq 1/4 + 1/20$. Suppose not.
Then, any clipping of $1/4$ of the probability mass of $\wvec{s}{\ell}$
leaves at least $1/20$ probability mass on $H$. The size of $H$ is at most $\ell^8$.
By Jensen's inequality, $\|\clip{\wvec{s}{\ell}}{1/4}\|^2_2 \geq 1/400\len^8$, contradicting
the case condition.

Thus, $\sum_{v \in H} \wvec{s}{\ell}(v) \leq 1/4 + 1/20$. The expected value
of $\sum_{v \in H} w_v \leq w(1/4 + 1/20)$. By an additive
Chernoff bound (first part, Theorem 1.1 of~\cite{DuPa-book}), $\Pr[\sum_{v \in H} w_v \geq w/3] \leq \exp(-2(1/3-1/4-1/20)^2w)
\leq \exp(-\len^{13})$. By \Clm{set}, with probability at least $1-2^{-1/\eps^2}$, 
$T \subseteq H$. By a union bound, with probability at least $1-2^{-1/\eps}$,
$\sum_{v \in T} w_v \leq \sum_{v \in H} w_v < w/3$, and the output is LOW.

Now for the second case. Let $H' = \{v \ | \wvec{s}{\ell}(v) \geq \len^{-7} \}$. 
We will show that $\sum_{v \in H} \wvec{s}{\ell}(v) \geq 3/8$. Suppose not.
We can clip away all the probability mass of $\wvec{s}{\ell}$ that is on $H$, which is at most $3/8$.
All remaining probability/entries of the clipped vector are at most $\len^{-7}$.
Thus, the squared $l_2$-norm is at most $\len^{-7}$, implying $\|\clip{\wvec{s}{\ell}}{3/8}\|^2_2 \leq \len^{-7}$ (contradiction).

Thus, $\sum_{v \in H'} \wvec{s}{\ell}(v) \geq 3/8$. 
By an additive Chernoff bound (first part,
Theorem 1.1 of~\cite{DuPa-book}), $\Pr[\sum_{v \in H} w_v < w/3] \leq \exp(-2(3/8-1/3)^2w) \leq \exp(-\len^{13})$.
By \Clm{set}, with probability at least $1-2^{-1/\eps^2}$, $H' \subseteq T$.
By a union bound, with probability at least $1-2^{-1/\eps}$,
$\sum_{v \in T} w_v \geq \sum_{v \in H'} w_v \geq w/3$, and the output is HIGH.
\end{proof}

We now prove completeness, \Thm{completeness}. We will prove that 
if $G$ is $H$-minor-free, then the tester \isminorfree{} accepts with probability $> 2/3$.
This follows almost directly from \Lem{clip:new}.

\begin{proof}[Proof of \Thm{completeness}] Suppose $G$ is $H$-minor-free. 
Note that calls to \localsearch{} can never return FOUND, so rejection
can only happen because of the output of calls to \estclip.

By \Lem{clip:new},
there are at least $(1-1/\len)n$ vertices such that $\|\clip{\wvec{s}{\ell}}{3/8}\|^2_2 \geq \len^{-7}$.
Call these vertices \emph{heavy}. The expected number of light vertices in the multiset
$S$ chosen in \Step{S} of \isminorfree{} is at most $1/\len \times \len^{21} = \len^{20}$.
By a multiplicative Chernoff bound (Theorem 1 of~\cite{DuPa-book}), the number
of light vertices in $S$ is strictly less than $2\len^{20}$ with probability
at least $1-\exp(-\len^{19}) > 9/10$. Let us condition on this event. 
The probability that any call to \estclip$(s)$ returns HIGH for a heavy $s \in S$
is at least $1-2^{-1/\eps}$, by \Clm{estclip}. By a union bound over the at most $\len^{21}$
heavy vertices in $S$, all calls to \estclip$(s)$ for heavy $s \in S$
return HIGH with probability at least $1-\len^{21}2^{-1/\eps} > 9/10$. 

We now remove the conditioning.
With probability $> (9/10)^2 > 2/3$, there are strictly less than $2\len^{18}$ calls (for the light vertices) that return
LOW. When this happens, \isminorfree{} accepts.
\end{proof}

Now we prove soundness, \Thm{soundness}. We prove that if $G$
is $\eps$-far from $H$-minor-freeness, the tester rejects with probability $> 2/3$.
The main ingredient is the decomposition of \Thm{partition}.

\begin{proof}[Proof of \Thm{soundness}] Assume $G$ is $\eps$-far from being $H$-minor free.
We split into two cases.

{\bf Case 1:} There are less than $(1-1/\len^{1/5})n$ vertices such that $\|\clip{\wvec{s}{\ell}}{1/4}\|^2_2 > \len^{-9}$.

Then, there are at least $n/\len^{1/5}$ vertices such that $\|\clip{\wvec{s}{\ell}}{1/4}\|^2_2 \leq \len^{-9}$.
The expected number of such vertices (with repetition) in the multiset $S$ (of \Step{S}) is at least
$\len^{21}/\len^{1/5}$. By a multiplicative Chernoff bound, there are at least $\len^{21}/2\len^{1/5} > 2\len^{20}$
such vertices in $S$, with probability at least $1-\exp(-\len^{20}/4)$. For each such vertex $s$,
the probability that \estclip$(s)$ outputs LOW is at least $1-2^{-1/\eps}$ (\Clm{estclip}). 
By a union bound over all vertices in $S$, with probability $> (1-\exp(-\len^{20}))(1-\len^{21}2^{-1/\eps}) > 5/6$,
there are at least $2\len^{20}$ calls to \estclip$(s)$ that return LOW. So the tester rejects.

{\bf Case 2:} There are at least $(1-1/\len^{1/5})n$ vertices such that $\|\clip{\wvec{s}{\ell}}{1/4}\|^2_2 > \len^{-9}$.

We apply the decomposition of \Thm{partition} (with $c=9$). There is a partition $\{P_1, P_2, \ldots, P_b\}$
of the vertices such that:
\begin{asparaitem}
    \item For each $P_i$, there exists $s \in V$ such that: $\forall v \in P_i$, $\sum_{t < 10\len^{10}} \prw{s}{v}{t} \geq 1/8\len^{10}$.
    Call $s$ the \emph{anchor} for $P_i$, noting that multiple sets may have the same anchor.
    \item The total number of edges crossing the partition is at most $24dn\sqrt{\len^{-1/5}\log \len}$.
\end{asparaitem}

\medskip

Among the sets in the partition, let $\{Q_1, Q_2, \ldots, Q_a\}$ be the sets of vertices
that contain an $H$-minor (or technically, the subgraphs induced by these sets contain an $H$-minor).
Note that one can remove $d\sum_{i \leq a} |Q_i| + 24dn\sqrt{\len^{-1/5}\log \len}$ edges
to make $G$ $H$-minor-free. Since $\len > \eps^{-20}$, $24dn\sqrt{\len^{-1/5}\log \len} \leq \eps nd/2$.
Since $G$ is $\eps$-far from being $H$-minor free, we deduce from the above 
that $\sum_{i \leq a} |Q_i| \geq \eps n/2$.

Let $Z = \{s \ | \textrm{$s$ is anchor for some $Q_i$}\}$. Let us lower bound $|Z|$. 
For every $Q_i$,
there is some $s \in Z$ such that $\forall v \in Q_i$, $\sum_{t < 10\len^{10}} \prw{s}{v}{t} \geq 1/8\len^{10}$.
Thus, for every $Q_i$, there is some $s \in Z$ such that $\sum_{v \in Q_i} \sum_{t < 10\len^{10}} \prw{s}{v}{t} 
\geq |Q_i|/8\len^{10}$. 
Let us sum over all $s \in Z$ (and note that $\sum_{v \in V} \prw{s}{v}{t} = 1$).
\begin{equation}
    \sum_{i \leq a} |Q_i|/8\len^{10} \leq \sum_{s \in Z} \sum_{v \in V} \sum_{t < 10\len^{10}} \prw{s}{v}{t}
    \leq \sum_{t < 10\len^{10}} \sum_{s \in Z} \sum_{v \in V} \prw{s}{v}{t} \leq 10\len^{10}|Z|
\end{equation}
Since $\sum_{i \leq a} |Q_i| \geq \eps n/2$, $|Z| \geq \eps n/160\len^{20} \geq 5n/\len^{21}$.

Focus on the multiset $S$ in \Step{S} of \isminorfree. Note that $S$ contains
an element of $Z$ with probability $\geq 1 - (1-5/\len^{21})^{\len^{21}} \geq 9/10$.
Let us condition of this event, and let $s \in S \cap Z$. There exists some $Q_i$
such that $\forall v \in Q_i$, $\sum_{t < 10\len^{10}} \prw{s}{v}{t} \geq 1/8\len^{10}$.
By averaging over walk length, $\forall v \in Q_i$, $\exists t < 10\len^{10}$ such
that $\prw{s}{v}{t} \geq 1/80\len^{20}$. 

Now, consider the call to \localsearch$(s)$. The set $B_s$ in \Step{walk} of \localsearch{}
is constructed by performing $\len^{21}$ random walks of length $\len^{11}$.
For any $v \in Q_i$, the probability that $v$ is in $B_s$
is at least $1 - (1-1/80\len^{20})^{\len^{21}} \geq 1 - \exp(-\len/80)$.
Taking a union bound over all $v \in Q_i$, the probability
that $Q_i \subseteq B_s$ is at least $1-\len^{21}\exp(-\len/80) \geq 9/10$.
When $Q_i \subseteq B_s$, then $G[B_s]$ contains an $H$-minor and 
the tester rejects. The probability of this happening is at least $(9/10)^2 > 2/3$.
\end{proof}

\bibliographystyle{alpha}
\bibliography{minor-freeness}

\begin{appendix}

\section{Local partitioning, and completing the proof of~\Lem{partitioning}}

We perform local partitioning on $M_S$, starting with an arbitrary $s \in S'$.
We apply the Lov\'{a}sz-Simonovits curve technique. (The definitions are originally from~\cite{LS:90}.
Refer to Lecture 7 of Spielman's notes~\cite{Sp-notes} as well as Section 2 in Spielman-Teng~\cite{ST12}.
This is also a restatement of material in Section 6.1 of~\cite{KSS:18-full}, which
is needed to state the main lemma.)

\begin{itemize}
    \item Conductance: for some $T \subseteq S$ we define the conductance of $T$ in $M_S$ to be 
    $$ \Phi(T) = \frac{\sum_{\substack{u \in T \\ v \in S \setminus T}}\projwp{u}{v}{1}} { \min\left\{|S\setminus T|, |T|\right\} }$$
    \item Ordering of states at time $t$: At time $t$,
	let us order the vertices in $M_S$ as $v^{(t)}_1, v^{(t)}_2, \ldots$ such that
	$\projwp{s}{v^{(t)}_1}{t} \geq \projwp{s}{v^{(t)}_2}{t} \ldots$, breaking ties by vertex id.
    At $t=0$, we set $\projwp{s}{s}{0} = 1$, and all other values to $0$.
    \item The LS curve $h_t$: We define a function $h_t:[0,|S|] \to [0,1]$ as follows. For every $k \in [|S|]$, set $h_t(k) = \sum_{j \leq k} \projwp{s}{v^{(t)}_j}{t}$. (Set $h_t(0) = 0$.)
	For every $x \in (k,k+1)$, we linearly interpolate to construct $h(x)$.
	Alternately, $h_t(x) = \max_{\vec{w} \in [0,1]^{|S|}, \|\vec{w}\|_1 = x} \sum_{v \in S} [\projwp{s}{v}{t} - 1/n]w_i$.
    \item Level sets: For $k \in [0,|S|]$, we define the $(k,t)$-level set, $L_{k,t}$ to be $\{v^{(t)}_1, v^{(t)}_2, \ldots, v^{(t)}_k\}$.
    The \emph{minimum probability} of $L_{k,t}$ denotes $\projwp{s}{v^{(t)}_k}{t}$.
\end{itemize}

The main lemma of Lov\'{a}sz-Simonovits is the following (Lemma 1.4 of \cite{LS:90}, also refer to Theorem 7.3.3 of Lecture 7
in~\cite{Sp-notes}).

\begin{lemma} \label{lem:ls} 
	For all $k$ and all $t$,
	$$ h_t(k) \leq \frac{1}{2}[h_{t-1}(k - 2\min(k,n-k)\Phi(L_{k,t})) + h_{t-1}(k + 2\min(k,n-k)\Phi(L_{k,t}))]$$
\end{lemma}

We employ this lemma to prove a condition of the level set conductances. An analogous lemma was proven in~\cite{KSS:18-full}
for specific parameters. We redo the calculation here.

\begin{lemma} \label{lem:level} 
	Suppose there exists $\phi \in [0, 1]$ and $p > 2/n$ such that for all $t' \leq t$ it is true that for all $k \in [n]$ that if $L_{k,t'}$ has a minimum
	probability of at least $p$, then $\Phi(L_{k,t}) \geq \phi$.
	Then for all $k \in [0,n]$, $h_t(k) \leq \sqrt{k}(1-\phi^2/2)^t + pk$.
\end{lemma}

\begin{proof} We will prove by induction over $t$. For the base case, consider $t=0$. The RHS
is at least $1$, proving the bound.

Now for the induction. Note that $h_t$ is a concave, and the RHS is also concave. Thus,
it suffices to prove the bound for the integer points ($h_t(k)$ for integer $k$).
If $k \geq 1/p$, then the RHS is at least 1. Thus the bound is trivially true. Let us assume that $k < 1/p< n/2$. We now split the proof into two cases based on the conductance of $L_{k, t}$.

	First let us consider the case where $\Phi(L_{k, t}) \geq \phi$. By \Lem{ls} and concavity of $h$,

	\begin{align}
		h_t(k) &\leq \frac{1}{2} \Big( h_{t-1} \big( k(1 - 2\phi)\big) + h_{t-1}\big(k (1 + 2\phi)\big)\Big) \\
		&\leq \frac{1}{2} \Big(  \sqrt{ k(1 - 2\phi)} (1 - \phi^2/2)^{t -1} +  \sqrt{ k(1 + 2\phi)} (1 - \phi^2/2)^{t -1} + 2kp \Big)\\
		&\leq \frac{1}{2} \Big( \sqrt{k} \left( 1 - \phi^2/2\right)^{t-1}\left( \sqrt{1 - 2\phi} + \sqrt{1 + 2\phi} \right)  + 2kp \Big) \\
		&\leq \sqrt{k}\left( 1 - \phi^2/2\right)^t + kp
	\end{align}
	\noindent For the last inequality we use the bound $\left( \sqrt{1 + z} + \sqrt{1 - z}\right) / 2 \leq 1 - z^2/8$.

Now we deal with the case when $\Phi(L_{k, t}) < \phi$. 
By assumption, $L_{k, t}$ has minimum probability less than $p$. Let $k' < k$ be the largest index such that $L_{k', t}$ has 
minimum probability at least $p$. Note that $\Phi(L_{k', t}) \geq \phi$. Therefore, as proven in the first case, $h_t(k') \leq \sqrt{k'}\left( 1 - \phi^2/2\right)^t + k'p$. Every vertex we add to $L_{k', t}$ adds less than $p$ probability mass to $L_{k', t}$, and therefore, by the concavity of $h_t(x)$,
	\begin{align}
		h_t(k) &\leq h_t(k') + (k - k')p \\
		&\leq \sqrt{k'}\left( 1 - \phi^2/2\right)^t + k'p + (k - k')p \\
		&\leq \sqrt{k'}\left(1 - \phi^2/2\right)^t + kp \leq \sqrt{k}\left(1 - \phi^2/2\right)^t + kp 
	\end{align}
\end{proof}

For convenience, we restate \Lem{partitioning}.

\partition*

\begin{proof} 
By \Lem{hop-infty-norm}, there is a set $S' \subseteq S$, $|S'| \geq |S|/2$ such that for all $s \in S'$, 
$\| \projw{s}{\len^{1/5}}\|_\infty \geq 1/2\len^{c+1}$. Consider any $s \in S'$.

Suppose for all $t' \leq \len^{1/5}$, all level sets $L_{k,t'}$ with minimum probability $1/2\len^{c+1}$
have conductance at least $\sqrt{4c\len^{-1/5}\log \len}$. \Lem{level} implies that
$\| \projw{s}{\len^{1/5}}\|_\infty = h_{\len^{1/5}}(1) \leq (1-2c\len^{-1/5}\log \len)^{\len^{1/5}} + 1/4\len^{c+1}$ 
$< 1/4\len^{c+1} + 1/4\len^{c+1}= 1/2\len^{c+1}$. This is a contradiction.

Thus, for every $s \in S'$, there exists a level set denoted $P_s$
with minimum probability $1/2\len^{c+1}$ and conductance at most $\sqrt{4c\len^{-1/5}\log \len}$.
Note that $|P_s| \leq 2\len^{c+1} < |S|/2$. 
\begin{equation}
\sqrt{4c\len^{-1/5}\log \ell} \geq 
\Phi(P_s) = \frac{\sum_{\substack{x \in P_s \\ y \in S \setminus P_s}} \projwp{x}{y}{1}}{\min(|P_s|, |S \setminus P_s|}
\geq \frac{E(P_s, S\setminus P_s)}{2d|P_s|}
\end{equation}

The inequality is obtained by only considering transitions from $S$ to $S \setminus P_s$ that come
from a single edge in $G$. Each such edge has a traversal probability of $1/2d$.
Therefore, $E(P_s, S\setminus P_s) \leq 4d|P_s|\sqrt{c\len^{-1/5}\log \len}$.

Set $L = 8\len^{c+2}$. Define $\tS \subseteq S'$ to be the vertices $s \in S'$
with the property that $\forall v \in P_s$, $\sum_{l < L} \wpr{s}{v}{l} \geq 1/8\len^{c+1}$.
Together with the cut bound above, this clearly satisfies the conditions
on the lemma. It remains the prove a suitable upper bound of $|S' \setminus \tS|$,
to show that $\tS$ is sufficiently large.

For every $s \in S' \setminus \tS$, there exists $v_s \in P_s$
such that $\sum_{l < L} \wpr{s}{l}{v} < 1/8\len^{c+1}$. Let $\projp{s}{l}{v}$
denote that probability that an $\len^{1/5}$-hop walk in $M_S$
from $s$ reaches $v$ with length $l$. Consider $s \in S' \setminus \tS$.
\begin{equation}
\projwp{s}{v_s}{\len^{1/5}} = \sum_{l \geq \len^{1/5}} \projp{s}{l}{v_s} 
= \sum_{l \geq \len^{1/5}}^{L-1} \projp{s}{l}{v_s} + \sum_{l \geq L} \projp{s}{l}{v_s}
\leq \sum_{l \geq \len^{1/5}}^{L-1} \wpr{s}{l}{v} + \sum_{l \geq L} \projp{s}{l}{v} 
\end{equation}
Since the minimum probability of $P_s$ is at least $1/4\len^{c+1}$,
$\projwp{s}{v_s}{\len^{1/5}} \geq 1/4\len^{c+1}$. 
We argued above that 
$\sum_{l \geq \len^{1/5}}^{L-1} \wpr{s}{l}{v}
\leq \sum_{l < L} \wpr{s}{l}{v} < 1/8\le^{c+1}$.
We conclude that $\sum_{l > L} \projp{s}{l}{v} \geq 1/8\len^{c+1}$.
Note that all of this probability mass corresponds to $\len^{1/5}$-hop walks that
have a large length.
We now lower bound $\EX_{W \sim \cW_{\len^{1/5}}}[\textrm{length of $W$}]$.
\begin{equation}
\EX_{W \sim \cW_{\len^{1/5}}}[\textrm{length of $W$}]
\geq \frac{1}{|S|}\sum_{s \in S' \setminus \tS}\Big(\sum_{l > L} \projp{s}{l}{v_s}\Big) L
\geq \frac{|S' \setminus \tS|}{|S|} \cdot \frac{L}{8\len^{c+1}} \geq \frac{\ell |S' \setminus \tS|}{|S|}
\end{equation}
By \Lem{kac}, $\EX_{W \sim \cW_{\len^{1/5}}}[\textrm{length of $W$}] = \len^{1/5}n/|S|$.
Combining, $|S' \setminus \tS| \leq n/\len^{4/5} \leq n/4\len^{1/5}
\leq |S|/4$. 
By \Lem{hop-infty-norm}, $|S'| \geq |S|/2$. By the setting of
\Lem{partitioning},
$|S| > n/\len^{1/5}$. Thus, $|S' \setminus \tS| \leq n/4\len^{1/5}$, and $|\tS| \geq |S|/4$.
\end{proof}

\end{appendix}

\end{document}